\documentclass[journal,twocolumn]{IEEEtran}

\usepackage[cmex10]{amsmath}
\usepackage{amsfonts,amssymb,amsbsy}
\usepackage{graphicx,comment}
\usepackage{tikz}
\usetikzlibrary{arrows}
\usepackage{pgfplots}
\pgfplotsset{compat=newest} 
\pgfplotsset{plot coordinates/math parser=false}
\usetikzlibrary{plotmarks}

\newtheorem{theorem}{Theorem}

\newtheorem{lemma}[theorem]{Lemma}

\newenvironment{proof}[1][Proof]{{\em #1:} }{\ \rule{0.5em}{0.5em}}

\newcommand{\set}[1]{\mathcal{#1}} 

\newcommand{\ul}[1]{\underline{#1}}
\newcommand{\E}[1]{{\rm E}\left[{#1}\right]}
\def\tg{\tilde{g}}
\def\tvol{\text{Vol}}
\def\tI{\tilde{I}}
\def\cP{\mathcal{P}}

\def\tC{\tilde{C}}
\def\tX{\tilde{X}}
\def\tx{\tilde{x}}
\def\re{\text{Re}}
\def\sinc{\text{sinc}}

\begin{document}
\title{Capacity Bounds for Discrete-Time, Amplitude-Constrained,
Additive White Gaussian Noise Channels}

\author{Andrew~Thangaraj,~Gerhard~Kramer~and~Georg~B\"ocherer
\thanks{A. Thangaraj is with the Department of Electrical Engineering,
  Indian Institute of Technology Madras,
Chennai 600036, India, Email: andrew@ee.iitm.ac.in}
\thanks{G. Kramer and G. B\"ocherer are with the Department of Electrical and Computer Engineering,
Technical University of Munich,
Arcisstra{\ss}e 21, 80333 M\"unchen, Germany, Email:
gerhard.kramer@tum.de, georg.boecherer@tum.de}
\thanks{This paper was presented in part at the IEEE International
  Symposium on Information Theory (ISIT) 2015, Hong Kong.}}


\maketitle

\begin{abstract}
The capacity-achieving input distribution of the discrete-time, additive white Gaussian noise
(AWGN) channel with an amplitude constraint is discrete and seems difficult
to characterize explicitly. A dual capacity expression is used to derive 
analytic capacity upper bounds for scalar and vector AWGN
channels. The scalar bound improves on McKellips' bound
and is within 0.1 bits of capacity for all signal-to-noise ratios (SNRs).
The two-dimensional bound is within 0.15 bits of capacity provably up to 4.5
dB, and numerical evidence suggests a similar gap for all
SNRs. 
\end{abstract}

\begin{keywords}
additive white Gaussian noise channel, amplitude constraint, capacity
\end{keywords}

\section{Introduction}
\label{sec:intro}
The most commonly-studied channel model for communications is the additive white
Gaussian noise (AWGN) channel. The AWGN model is interesting only with
constraints on the channel input or output. Depending on the application, one
is interested in limiting, e.g., the average input (or output) variance or the
input amplitude.

Input (or output) variance constraints result in elegant analytic capacity
expressions such as Shannon's $\frac{1}{2}\log(1+\text{SNR})$ formula. The
amplitude constraint seems less tractable, and typical analyses use Smith's
methods~\cite{Smith1971} to show that the capacity-achieving input distribution
has discrete amplitudes, see \cite{Shamai95,Tchamkerten04,Chan05,Sharma08}
and references therein. A recent line of work studies the peak-to-average power
(PAPR) ratio of good codes~\cite{PolyanskiyWu14}.

An alternative approach is by McKellips~\cite{McKellips:04isit} who develops
analytic and tight capacity upper bounds by bounding the channel output entropy.
We instead use the dual capacity expression in~\cite[p.~128]{Csiszar11}
(see also~\cite[Eqn.~(7)]{Lapidoth:03}) and study both scalar and vector channels.
The dual approach was also used in \cite{LMW2009} for scalar AWGN channels
with non-negative channel inputs.  Our models and results differ from those in \cite{LMW2009}:\
we do not impose a non-negativity constraint (this difference turns out to be minor for the
scalar case), we study vector channels that include the important practical case of
two-dimensional (complex) AWGN channels, and we develop certain bounds in more detail.
Our bounds are within 0.15 bits of capacity provably up to 4.5 dB, and numerically for all SNRs
for two-dimensional (complex) AWGN channels.

This paper is organized as follows.
Section~\ref{sec:prelim} presents functions, integrals, and bounds that we need later.
Sections~\ref{sec:1D}-\ref{sec:nD} develop the one-, two-, and $n$-dimensional
bounds, as well as two refinements. The appendices contain technical proofs.

\section{Preliminaries}
\label{sec:prelim}
Consider the following functions:
\begin{align*}
& \psi(x) \overset{(a)}{=} \frac{1}{\sqrt{2\pi}} e^{-x^2/2} \\
& Q(x) \overset{(b)}{=} \int_{x}^{\infty} \psi(z) \: dz \\
& I_{0}(x) \overset{(c)}{=} \frac{1}{\pi} \int_{0}^{\pi} e^{x \cos \phi} \: d\phi \\
& Q_{1}(a,b) \overset{(d)}{=} \int_{b}^{\infty} z \: e^{-(z^2+a^2)/2} I_{0}(az) \: dz \\
& D(p\|q) \overset{(e)}{=} \int_{\infty}^{\infty} p(z) \log \frac{p(z)}{q(z)} \: dz
\end{align*}
where $(a)$ is the Gaussian density, $(b)$ is the Q-function,
$(c)$ is the modified Bessel function of the first kind of integer order $0$, 
$(d)$ is the Marcum Q-function,
and $(e)$ is the informational divergence between the densities $p$
and $q$. Logarithms to the base $e$ and base 2 are denoted as $\log$ and $\log_2$, respectively. 
A few useful properties are $Q_{1}(a,0)=1$ and the bounds
\begin{align}
& \frac{x}{1+x^2} \, \psi(x) < Q(x) < \frac{1}{x} \, \psi(x) \label{eq:Q-bounds}
\end{align}
for $x>0$ (and for $x=0$). We also consider the integrals:
\begin{align}
& \int_x^{\infty} y \, \psi(y) \, dy = \left. -\psi(y)\right|_x^{\infty} = \psi(x) \label{eq:psi-integral-1} \\
& \int_x^{\infty} y^2 \, \psi(y) \, dy = \int_x^{\infty} y (-d\psi(y)) = x \psi(x) + Q(x) \label{eq:psi-integral-2} \\
& \int_x^{\infty} y^3 \, \psi(y) \, dy = \int_x^{\infty} y^2 (-d\psi(y)) = (x^2 + 2) \psi(x) . \label{eq:psi-integral-3} 
\end{align}
For sequences $f(n)$ and $g(n)$, the standard big-$O$ notation
$f(n)=O(g(n))$ denotes that $|f(n)|\le c|g(n)|$ for a constant $c$ and
sufficiently large $n$ \cite{Cormen:2001:IA:580470}. 
Finally, a useful upper bound on the capacity $C$ of a memoryless channel $p_{Y|X}(\cdot)$
is based on the dual capacity expression \cite[p.~128]{Csiszar11}, \cite[Eqn.~(7)]{Lapidoth:03}.
The bound is
\begin{align} \label{eq:upper-bound}
   C \le \max_{x \in \set{S}} \: D\left( \, p_{Y|X}(\cdot | x) \, \| \, q_Y(\cdot) \, \right)
\end{align}
where $q_{Y}(\cdot)$ is any choice of ``test" density $q_{Y}(\cdot)$ and $\set{S}$
is the set of permitted $x$. 

\section{Real AWGN Channel}
\label{sec:1D}
Consider the real-alphabet AWGN channel
\begin{align}
   Y = X + Z
\end{align}
where $Z$ is a Gaussian random variable with mean $0$ and variance $1$.
The channel density is
\begin{align}
   p_{Y|X}(y|x) = \psi(y-x).
\end{align}

Consider the amplitude constraint $|X| \le A$ where $A>0$.
We choose a family of test densities
\begin{align} \label{eq:test-density-1D}
   q_{Y}(y) = \left\{
   \begin{array}{ll}
      \frac{\beta}{2A}, & |y| \le A \\
      \frac{1-\beta}{\sqrt{2\pi}} e^{-(y-A)^2/2}, & |y|>A
   \end{array} \right.
\end{align}
where $\beta\in[0,1]$ is a parameter to be optimized.
The test density is illustrated in Fig. \ref{fig:testreal}.
It is a mixture of two distributions,
a uniform distribution in the interval $|y|\le A$ and
a ``split and scaled" Gaussian density for $|y|> A$. The parameter
$\beta$ specifies the mixing proportion.

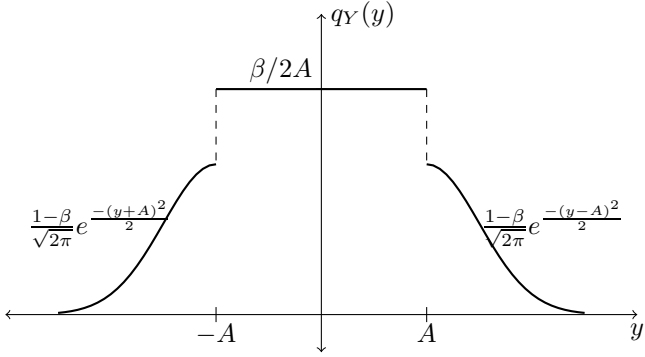
\begin{figure}[t!]
  \centering
  \begin{tikzpicture}[xscale=0.7,yscale=1]
\draw[<->] (-6,0) -- (6,0);
\draw[<->] (0,4) -- (0,-0.5);
\draw (-2,0.1) -- (-2,-0.1);
\draw (2,0.1) -- (2,-0.1);
\draw (0.1,3) -- (-0.1,3);
\node[left] at (0,3.25) {$\beta/2A$};
\node[below] at (2,0) {$A$};
\node[below] at (-2,0) {$-A$};
\node[right] at (0,4) {$q_Y(y)$};
\node[below] at (6,0) {$y$};
\draw[domain=-2:2,thick] plot (\x,3);
\draw[dashed] (-2,3) -- (-2,2);
\draw[dashed] (2,3) -- (2,2);
\draw[domain=-5:-2,thick] plot(\x,{2*exp(-(\x+2)*(\x+2)/2)});
\draw[domain=2:5,thick] plot(\x,{2*exp(-(\x-2)*(\x-2)/2)});
\node at (-4.2,1.2) {$\frac{1-\beta}{\sqrt{2\pi}}e^{\frac{-(y+A)^2}{2}}$};
\node at (4.4,1.2)
{$\frac{1-\beta}{\sqrt{2\pi}}e^{\frac{-(y-A)^2}{2}}$};
\end{tikzpicture}
  \caption{Test densities for the real AWGN channel.}
  \label{fig:testreal}
\end{figure}

Inserting \eqref{eq:test-density-1D} into the divergence in \eqref{eq:upper-bound}, we have
\begin{align} 
   & D\left( p_{Y|X}(\cdot | x) \| q_Y(\cdot) \right) \nonumber \\
   & = \int_{-\infty}^{\infty} p_{Y|X}(y | x) \log \frac{p_{Y|X}(y | x)}{q_{Y}(y)}  \: dy \nonumber \\
   & = -\log\left(\sqrt{2 \pi e}\right) - \log\left(\frac{\beta}{2A}\right) \left[ 1 - Q(A+x) - Q(A-x) \right] \nonumber \\
   & \quad - \log\left(\frac{1-\beta}{\sqrt{2 \pi}}\right) \left[ Q(A+x) + Q(A-x) \right] \nonumber \\
   & \quad + \frac{1}{2} \left\{ [(A+x)^2 + 1] Q(A+x) - (A+x) \psi(A+x) \right. \nonumber \\
   & \quad \; \left. + [(A-x)^2 + 1] Q(A-x) - (A-x) \psi(A-x) \right\}
     \nonumber\\
&=\log\frac{2A}{\beta\sqrt{2\pi e}} +\log\frac{\beta\sqrt{2\pi e}}{(1-\beta)2A}\left[ Q(A-x) +
  Q(A+x) \right]\nonumber \\
&\quad +\frac{1}{2}\left[g(A-x)+g(A+x)\right] \label{eq:D-bound1}
\end{align}
where $g(u)\triangleq u^2Q(u)-u\psi(u)$. 

\subsection{McKellips' bound}
Using \eqref{eq:Q-bounds}, we
readily see that $g(u)\le0$ for $u>0$. By symmetry, we may restrict
attention to $0 \le x \le A$ so that
$g(A-x)+g(A+x)\le0$. Using \eqref{eq:D-bound1}, we thus have
\begin{align}
   D( p_{Y|X}&(\cdot | x) \| q_Y(\cdot))\le \log\frac{2A}{\beta\sqrt{2\pi e}} \nonumber\\
 &+\log\frac{\beta\sqrt{2\pi e}}{(1-\beta)2A}\left[ Q(A-x) + Q(A+x) \right].   \label{eq:9}
\end{align}
To recover McKellips' bound \cite{McKellips:04isit}, we choose
\begin{align}
   \beta = \frac{2A}{\sqrt{2 \pi e}+2A}\label{eq:beta-choice}
\end{align}
to make the second term in \eqref{eq:9} equal to zero,
and obtain
\begin{align} \label{eq:C-bound-1}
   C \le \log\left( 1 + \frac{2A}{\sqrt{2 \pi e}}  \right).
\end{align}
We now combine \eqref{eq:C-bound-1} with the capacity under the
(weaker) average power constraint $\E{X^2} \le A^2$.
The noise power is 1 so the signal-to-noise ratio (SNR) is $P=\sqrt{A}$.
We thus arrive at McKellips' bound in~\cite{McKellips:04isit}:
\begin{align} \label{eq:C-bound-1-McKellips}
   C \le \min \left\{ \log\left( 1 + \sqrt{\frac{2P}{\pi e}}  \right), \frac{1}{2} \log\left( 1+P \right) \right\} .
\end{align}

Observe that the high-SNR power loss is $10\log_{10}(\pi e/2) \approx 6.30$ dB. 
However, this comparison is based on equating the maximum power $P$ with the
average power. Instead, if we use the uniform distribution for $X$ then the average
power is $P/3$ and the high-SNR power loss reduces to the high-SNR shaping
loss of $10\log_{10}(\pi e/6) \approx 1.53$ dB (see~\cite{OzarowWyner90}).

\subsection{Refined Bound}
McKellips' bound seems tight for high SNR,
roughly above 6 dB. For low SNR, below 0 dB, the
average-constraint bound $\frac{1}{2}\log(1+P)$ is tight. For the
intermediate range between 0 to 6 dB, we derive a better bound
next.

Consider $\beta$ for which
$\log\frac{\beta\sqrt{2\pi e}}{(1-\beta)2A}$ is positive, i.e., consider the range
\begin{equation}
  \label{eq:12}
  \beta\ge\frac{2A}{\sqrt{2\pi e}+2A}.
\end{equation}
Observe that $Q(A-x)+Q(A+x)$ increases with $x$ if
$x\in[0,A]$, since the derivative evaluates to
$\psi(A-x)-\psi(A+x)$ which is positive for $x\in[0,A]$. Thus, the RHS of
\eqref{eq:9} is maximized by setting $x=A$, and we obtain the bound
\begin{align}
  &D( p_{Y|X}(\cdot | x) \| q_Y(\cdot))\le
    \log\frac{2A}{\beta\sqrt{2\pi e}} \nonumber\\
&\phantom{D( p_{Y|X}(\cdot | x) \| q_Y(\cdot))}+\log\frac{\beta\sqrt{2\pi e}}{(1-\beta)2A}\left[ 1/2 + Q(2A)
 \right]\nonumber\\
&\quad =(1/2-Q(2A)) \log\frac{2A}{\sqrt{2\pi e}}-(1/2-Q(2A))\log\beta\nonumber\\
&\phantom{-(1/2-Q(2A))\log\beta}-(1/2+Q(2A))\log(1-\beta).   \label{eq:13}
\end{align}
Setting $\beta=1/2-Q(2A)$, minimizes the RHS of \eqref{eq:13}. However, from
\eqref{eq:12} this choice of $\beta$ is valid only if
\begin{equation}
  \label{eq:15}
  1/2-Q(2A)\ge \frac{2A}{\sqrt{2\pi e}+2A}
\end{equation}
which is equivalent to $A\le 2.0662\approx \sqrt{\pi
  e/2}$. Therefore, using $P=\sqrt{A}$, we have the bound
\begin{equation}
  \label{eq:14}
  C\le \beta(P) \log\sqrt{\frac{2P}{\pi e}}+H_e(\beta(P)),\quad P\le 6.303 \text{dB}
\end{equation}
where $\beta(P)=1/2-Q(2\sqrt{P})$ and $H_e(x)=-x\log(x)-(1-x)\log(x)$
is the binary entropy function with the units of nats. 

\begin{figure}[t!]
  \centering
%
%
\begin{tikzpicture}

\begin{axis}[%
width=2.6in,
scale only axis,
xmin=-1,
xmax=10,
xlabel={SNR (dB)},
xmajorgrids,
ymin=0.4,
ymax=1.5,
ylabel={Capacity Bounds},
ymajorgrids,
legend style={at={(0.03,0.97)},anchor=north west,draw=black,fill=white,legend cell align=left}
]
\addplot [color=red,solid,very thick]
  table[row sep=crcr]{%
-6	0.161369636200964\\
-5	0.197726912233456\\
-4	0.240747161361905\\
-3	0.291036421211282\\
-2	0.34887885452954\\
-1	0.414111415529076\\
0	0.485944154132943\\
1	0.562788137758867\\
2	0.642148645592367\\
3	0.720660888666061\\
4	0.794353416651163\\
5	0.8628\\
6	0.9402\\
7	1.0275\\
8	1.1220\\
9	1.21025480770124\\
10	1.30882500120638\\
11	1.39559002360505\\
12	1.46620597777841\\
13	1.63275472070543\\
14	1.7428163567042\\
15	1.86964626331398\\
};
\addlegendentry{Lower bound};

\addplot [color=blue,dashed,thick]
  table[row sep=crcr]{%
-6	0.161649661346921\\
-5	0.198204580581557\\
-4	0.24173747667284\\
-3	0.293051963222674\\
-2	0.352859525386757\\
-1	0.421721912601826\\
0	0.5\\
1	0.587818317346194\\
2	0.685052334875493\\
3	0.791341177455778\\
4	0.906123095650313\\
5	1.0286866043034\\
6	1.15822808981313\\
7	1.29390718678102\\
8	1.43489360958514\\
9	1.58040221195651\\
10	1.72971580931865\\
11	1.88219718352143\\
12	2.03729261745271\\
13	2.19452948368152\\
14	2.35351013136442\\
15	2.51390383667526\\
};
\addlegendentry{$1/2\log(1+\text{SNR})$};

\addplot [color=blue,solid,mark=o,mark size=2.5,mark options={solid},thick]
  table[row sep=crcr]{%
-6	0.313298421918392\\
-5	0.347257761813228\\
-4	0.384432718621389\\
-3	0.425034177484001\\
-2	0.469269555411963\\
-1	0.51733946820157\\
0	0.569434169554036\\
1	0.625729868184064\\
2	0.686385051393113\\
3	0.751536960538134\\
4	0.821298372427153\\
5	0.895754838828694\\
6	0.974962522835825\\
7	1.0589467458713\\
8	1.14770132414599\\
9	1.24118873118819\\
10	1.33934107743872\\
11	1.44206185313117\\
12	1.54922834090375\\
13	1.66069457327432\\
14	1.77629468956403\\
15	1.89584653801677\\
};
\addlegendentry{McKellips};

\addplot [color=blue,solid,mark=diamond,mark size=3.5,mark options={solid},thick]
  table[row sep=crcr]{%
-6     0.2279\\
-5     0.2564\\
-4     0.2903\\
-3     0.3307\\
-2     0.3786\\
-1     0.4347\\
0      0.4988\\
1      0.5700\\
2      0.6467\\
3      0.7269\\
4      0.8089\\
5      0.8917\\
6      0.9747\\
};
\addlegendentry{Refined};

\end{axis}
\end{tikzpicture}%
  \caption{Capacity bounds for scalar AWGN channels. The rate units are bits per channel use.}
  \label{fig:1D}
\end{figure}
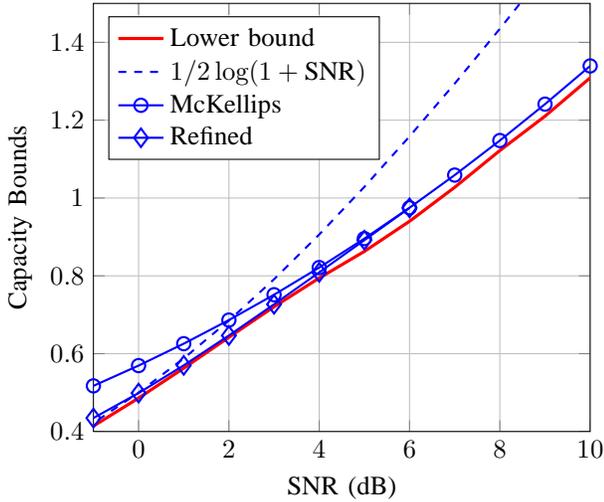
%
The bounds are plotted in Fig. \ref{fig:1D}, where the lower bound is taken
from~\cite{McKellips:04isit} with optimized input
distribution. The refined bound is best for
SNR from 0 to 5 dB but is valid only for SNR below 6.3 dB.

\section{Complex AWGN Channel}
\label{sec:2D}
Consider next the complex-alphabet AWGN channel
\begin{align}
   Y = X + Z
\end{align}
where $Z=Z_R + j Z_I$, $j=\sqrt{-1}$, and $Z_R$, $Z_I$ are independent Gaussian
random variables with mean $0$ and variance $1$.
The channel density in Cartesian coordinates with $x=[x_R,x_I]$ and
$y=[y_R,y_I]$ is
\begin{align}
   p_{Y|X}(y|x) = \frac{1}{2 \pi} e^{-\| y-x\|^2/2}.
\end{align}
In spherical coordinates with $x=[|x|,\phi_x]$ and $y=[|y|,\phi_y]$, the
density is 
\begin{align}
   p_{Y|X}(y|x) = \frac{1}{2 \pi} e^{-\left( |y|^2+|x|^2-2 |x| |y| \cos(\phi_y-\phi_x) \right)/2}.
\end{align}

Consider again the peak power constraint $|X| \le A$ where $A>0$.
We choose the test density
\begin{align} \label{eq:test-density-2D}
   q_{Y}(y) = \left\{
   \begin{array}{ll}
      \frac{\beta}{\pi A^2}, & |y| \le A \\
      \frac{1-\beta}{2\pi \left(1+\sqrt{\pi/2}A\right)} e^{-(|y|-A)^2/2}, & |y|>A.
   \end{array} \right.
\end{align}
Again, the test density is uniform in the interval $|y|\le A$ and
is a ``split and scaled" Gaussian density for $|y|> A$.
Inserting into the divergence \eqref{eq:upper-bound}, we
have\footnote{By symmetry we may restrict attention to $\phi_x=0$, i.e., real $x$ satisfying $0 \le x \le A$.}
\begin{align} 
   & D\left( p_{Y|X}(\cdot | x) \| q_{Y}(\cdot) \right)
      = -\log(2 \pi e) - \E{\log q_{Y}(Y)} \nonumber \\
   & = \log \frac{\pi A^2}{2\pi e \beta} - \E{\log\left( q_{Y}(Y) \frac{\pi A^2}{\beta} \right)} \nonumber \\
   & = \log \frac{A^2}{2 e \beta}  - \int_A^{\infty} e^{-(z^2+|x|^2)/2} I_0(z |x|) \nonumber \\ 
   & \quad \left[ \log\left(\frac{(1-\beta) \pi A^2}{2\pi
     \left(1+\sqrt{\pi/2}A\right) \beta }\right) - \frac{(z-A)^2}{2}
     \right] z \: dz.\label{eq:44}\\
   & = \log \frac{A^2}{2 e \beta} +\log\left(\frac {2e\left(1+\sqrt{\pi/2}A\right) \beta
     }{(1-\beta) A^2}\right)Q_1(|x|,A)\nonumber\\ 
&\quad\quad\quad\quad\quad-\tg(|x|,A) \label{eq:D-bound2}
\end{align}
where
\begin{equation}
  \label{eq:17}
  \tg(|x|,A)=\int_A^{\infty} e^{-(z^2+|x|^2)/2} I_0(z |x|)\left[1 - \frac{(z-A)^2}{2}\right] z \: dz.
\end{equation}
\begin{lemma}
$\tg(|x|,A)$ in \eqref{eq:17} is positive for
$|x|\in[0,A]$.  
\label{lem:tgcomplex}
\end{lemma}
\begin{proof}
 We use the definition of $I_0(x)$ to re-write \eqref{eq:17} as
\begin{align} \label{eq:D-bound2b}
   & \frac{1}{\sqrt{2\pi}} \int_0^{\pi} e^{-(|x|^2/2) \sin^2 \phi} \nonumber \\
   & \left[ \int_A^{\infty} \left[ 2 - (z-A)^2 \right] z \: \psi\left( z-|x|\cos\phi \right) \, dz\right] d\phi .
\end{align}
The integral in square brackets can be simplified by substituting
$\tilde{z}=z-|x|\cos\phi$ and $u=A-|x|\cos\phi$ to become
\begin{align} \label{eq:D-bound2c}
   \int_{u}^{\infty} \left[ 2 - (\tilde{z}-u)^2 \right] (\tilde{z}-u+A) \, \psi\left( \tilde{z} \right) \, d\tilde{z} .
\end{align}
Using \eqref{eq:psi-integral-1}-\eqref{eq:psi-integral-3}, the integral \eqref{eq:D-bound2c} evaluates to
\begin{align} 
   u(A-u) \left[ \psi(u) - u \, Q(u) \right] + (A+u)Q(u) \label{eq:D-bound2d}
\end{align}
or alternatively to
\begin{align}
   -(A-u) \left[ (1+u^2) Q(u) - u \, \psi(u) \right] + 2AQ(u) . \label{eq:D-bound2e}
\end{align}
Note that we have $0 \le u \le A+|x|$. We consider two cases.
\begin{itemize}
\item $0 \le u \le A$:
We have $u(A-u) \ge 0$ and the bound on the right-hand side of \eqref{eq:Q-bounds}
tells us that \eqref{eq:D-bound2d} is negative.
\item $A \le u \le A+|x|$:
We have $A-u \le 0$ and the bound on the left-hand side of \eqref{eq:Q-bounds}
tells us that \eqref{eq:D-bound2e} is positive.
\end{itemize}
Thus, the expression \eqref{eq:D-bound2d} (or equivalently \eqref{eq:D-bound2e}) is positive.
But this implies that the integrals in
\eqref{eq:D-bound2b}-\eqref{eq:D-bound2c} are all positive, and we
conclude that $\tg(A,x)$ is positive.
\end{proof}

\subsection{McKellips-type Bound}
Using Lemma \ref{lem:tgcomplex} in \eqref{eq:D-bound2}, we have
\begin{align}
  \label{eq:16}
  & D\left( p_{Y|X}(\cdot | x) \| q_{Y}(\cdot) \right)\nonumber\\
&\le \log \frac{A^2}{2 e \beta} +\log\left(\frac {2e\left(1+\sqrt{\pi/2}A\right) \beta }{(1-\beta) A^2}\right)Q_1(|x|,A).
\end{align}
By choosing $\beta$ to make the second term above zero, we have
\begin{align} \label{eq:beta-choice-2}
   \beta = \frac{A^2}{A^2 + 2 e \left(1+\sqrt{\pi/2}A\right)}
\end{align}
which results in the McKellips-type bound
\begin{align} \label{eq:D-bound2a}
   & D\left( p_{Y|X}(\cdot | x ) \| q_{Y}(\cdot) \right) \le \log\left( 1 +\sqrt{\pi/2}A + \frac{A^2}{2 e}  \right).
\end{align}
Since this bound is independent of $x$, we have
\begin{align} \label{eq:C-bound-2}
   C < \log\left( 1 +\sqrt{\pi/2}A + \frac{A^2}{2 e}  \right).
\end{align}
We combine \eqref{eq:C-bound-2} with the capacity under the
(weaker) average power constraint $\E{|X|^2} \le A^2$. Observe
that the complex noise has power 2 so the corresponding SNR
is $P=A^2/2$. We thus have the simple bound
\begin{align} \label{eq:C-bound-1-simple}
   C \le \min \left\{ \log \left( 1 +\sqrt{\pi P} + \frac{P}{e}  \right), \log\left( 1+P \right) \right\}
\end{align}
where we measure the rate per complex symbol (two real dimensions).

The high-SNR power loss is $10\log_{10}(e) \approx 4.34$ dB. 
Again, however, this comparison is based on equating the maximum power $P$ with the
average power. Instead, if we use the uniform distribution for $X$ then the average
power is $P/2$ and the high-SNR power loss reduces to a shaping loss of
$10\log_{10}(e/2) \approx 1.33$ dB.

\subsection{Refined Bound}
We refine the upper bound for the complex AWGN channel in a manner similar to the refinement in
the real AWGN case. First, rewrite the final expression for $D\left(
  p_{Y|X}(\cdot | x) \| q_Y(\cdot) \right)$ in \eqref{eq:44} as
\begin{align} 
   & D\left( p_{Y|X}(\cdot | x) \| q_Y(\cdot) \right) = \log\left(\frac{A^2}{2 e \beta}\right) \nonumber\\
  & + \log\frac{2\beta(1+\sqrt{\pi/2}A)}{(1-\beta)A^2} Q_1(|x|,A)
    +g(|x|,A) \label{eq:D-2}
\end{align}
where 
\begin{align*}
g(|x|,A)=\int_A^{\infty} \frac{(z-A)^2}{2} \, z \,
e^{-(z^2+|x|^2)/2} I_0(z |x|) \: dz. 
\end{align*}
The functions $Q_1(|x|,A)$ and $g(|x|,A)$ are both increasing in
$|x|$. This is proved as part of the general $n$-dimensional case in
Appendix B. Hence, for a positive
$\log\frac{2\beta(1+\sqrt{\pi/2}A)}{(1-\beta)A^2}$ the expression
\eqref{eq:D-2} is maximized at $|x|=A$, and we obtain
\begin{align}
   & D\left( p_{Y|X}(\cdot | x) \| q_Y(\cdot) \right) \le \log\left(\frac{A^2}{2 e \beta}\right) \nonumber\\
  & + \log\frac{2\beta(1+\sqrt{\pi/2}A)}{(1-\beta)A^2} Q_1(A,A)
    +g(A,A)  \label{eq:18}
\end{align}
provided that
\begin{equation}
  \label{eq:19}
\beta\ge\frac{A^2}{A^2+2(1+\sqrt{\pi/2}A)}.
\end{equation}
Rewriting the bound of \eqref{eq:18} as
\begin{align}
   & D\left( p_{Y|X}(\cdot | x) \| q_Y(\cdot) \right) \le
     Q_1(A,A)\log\frac{1+\sqrt{\pi/2}A}{e} \nonumber\\
&+(1-Q_1(A,A))\log\frac{A^2}{2e} +g(A,A)\nonumber\\
&-(1-Q_1(A,A))\log\beta-Q_1(A,A)\log(1-\beta)
  \label{eq:20}
\end{align}
we see that $\beta=1-Q_1(A,A)$ minimizes the RHS of \eqref{eq:20}.
However, from \eqref{eq:19} this choice of $\beta$ is valid only if
\begin{align}
  \label{eq:21}
  1-Q_1(A,A)\ge\frac{A^2}{A^2+2(1+\sqrt{\pi/2}A)}
\end{align}
which requires $A<2.36$ numerically. Therefore, setting
$P=A^2/2$ we have
\begin{align}
C\le &(1-\beta(P))\log(1+\sqrt{\pi P})+\beta(P)\log\frac{P}{e}\nonumber\\
&+H_e(\beta(P)) -\tg(\sqrt{2P},\sqrt{2P}),\quad P\le 4.45 \text{dB}
  \label{eq:22}
\end{align}
where $\beta(P)=1-Q_1(\sqrt{2P},\sqrt{2P})$ and we have used the
relationship $\tg(|x|,A)=Q_1(|x|,A)-g(|x|,A)$.
%



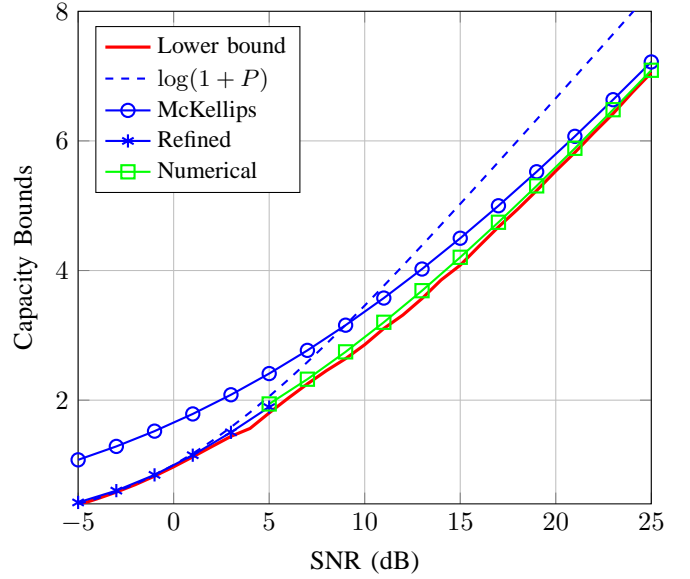
\begin{figure}[t!]
  \centering
%
%
\begin{tikzpicture}

\begin{axis}[%
width=3in,
height=2.58031536638153in,
scale only axis,
xmin=-5,
xmax=25,
xlabel={$\text{SNR}$ (dB)},
xmajorgrids,
ymin=0.4,
ymax=8,
ylabel={Capacity Bounds},
ymajorgrids,
legend style={at={(0.03,0.97)},anchor=north west,draw=black,fill=white,legend cell align=left,font=\small}
]
\addplot [color=red,solid,very thick]
  table[row sep=crcr]{%
-5	0.384281592903321\\
-4	0.478524051038729\\
-3	0.58149472201081\\
-2	0.697682594040937\\
-1	0.828216965657955\\
0	0.971888071362703\\
1	1.12557612261591\\
2	1.28429728960229\\
3	1.44132178162649\\
4	1.56400842299929\\
5	1.80319597687237\\
6	2.03183943337364\\
7	2.25059473509028\\
8	2.45583214510868\\
9	2.64047716777868\\
10	2.85572721331768\\
11	3.10590300772958\\
12	3.31401908890385\\
13	3.56779023432795\\
14	3.85445320828492\\
15	4.082307490959\\
16	4.38068569316326\\
17	4.67233390223964\\
18	4.94444707881743\\
19	5.23336692176461\\
20	5.53625977094028\\
21	5.82177319312868\\
22	6.12793098947061\\
23	6.42078938398519\\
24	6.74638768329348\\
25	7.05289299248812\\
};
\addlegendentry{Lower bound};

\addplot [color=blue,dashed,thick]
  table[row sep=crcr]{%
-5	0.396409161163114\\
-4	0.48347495334568\\
-3	0.586103926445348\\
-2	0.705719050773513\\
-1	0.843443825203652\\
0	1\\
1	1.17563663469239\\
2	1.37010466975099\\
3	1.58268235491156\\
4	1.81224619130063\\
5	2.0573732086068\\
6	2.31645617962626\\
7	2.58781437356203\\
8	2.86978721917029\\
9	3.16080442391302\\
10	3.4594316186373\\
11	3.76439436704286\\
12	4.07458523490543\\
13	4.38905896736305\\
14	4.70702026272884\\
15	5.02780767335052\\
16	5.3508761542486\\
17	5.67577990180488\\
18	6.00215644400198\\
19	6.32971245944191\\
20	6.65821148275179\\
21	6.98746345895592\\
22	7.31731600193655\\
23	7.64764716249041\\
24	7.97835949780125\\
25	8.30937524121281\\
};
\addlegendentry{$\log(1+P)$};

\addplot [color=blue,solid,mark=o,mark size=2.5,mark options={solid},thick]
  table[row sep=crcr]{%
-5	1.07933219075026\\
-3	1.28639487240923\\
-1	1.52201465010971\\
1	1.78737341273123\\
3	2.08327246933848\\
5	2.41013894129156\\
7	2.76803495388981\\
9	3.15665936362013\\
11	3.57534033112611\\
13	4.0230264535605\\
15	4.49828979794224\\
17	4.99935335913516\\
19	5.52414877157914\\
21	6.07040085647899\\
23	6.63572790845897\\
25	7.21774328342879\\
};
\addlegendentry{McKellips};

\addplot [color=blue,solid,mark=asterisk,mark size=2.5,mark options={solid},thick]
  table[row sep=crcr]{%
-5	0.42177751187092\\
-3	0.60413802000342\\
-1	0.847271527164506\\
1	1.15103777121458\\
3	1.50403038169789\\
5	1.89306204537921\\
};
\addlegendentry{Refined};

\addplot [color=green,solid,mark=square,mark size=2.5,mark options={solid},thick]
  table[row sep=crcr]{%
5	1.94141951568341\\
7	2.32321172560059\\
9	2.74535520921108\\
11	3.2016677092876\\
13	3.68901616765074\\
15	4.20434829103437\\
17	4.74441500302557\\
19	5.3059618276815\\
21	5.88587922177028\\
23	6.48130331277347\\
25	7.08967245092541\\
};
\addlegendentry{Numerical};

\end{axis}
\end{tikzpicture}%
  \caption{Capacity bounds for complex AWGN channels. The rate units
    are bits per 2 dimensions.}
  \label{fig:2D}
\end{figure}
%
The bounds are plotted in Fig. \ref{fig:2D}. The lower bound is
obtained by evaluating mutual information for the equi-probable complex
constellation
\begin{align}
\{0\}\;\cup\bigcup_{0\le k\le
  \lfloor A/2\rfloor-1}\{(A-2k)e^{jn\theta_k}:0\le n\le N_k\}
\end{align}
where $\theta_k=\frac{2\pi}{3(A-2k)}$ and $N_k=\lfloor
3(A-2k)\rfloor$. We see that the refined upper bound, valid for SNR
less than 4.45 dB, is close to the lower bound. The numerical
evaluation of $\min_{\beta}\max_x D\left( p_{Y|X}(\cdot | x) \|
  q_Y(\cdot) \right)$ is seen to yield best bounds throughout. 

\subsection{Analytical lower bound}
A lower bound on the capacity of real AWGN channels was derived in
\cite{OzarowWyner90} by using PAM-like constellations.
The input was peak-power constrained. By selecting the number of points suitably,
PAM achieves rates within
a small gap from the average-power constrained capacity
$\frac{1}{2}\log(1+\text{SNR})$.

For two-dimensions, consider the constellation
\begin{equation}
  \label{eq:52}
  A_N=\{0\}\cup\bigcup_{n=1}^{N-1}\{(n+0.5)\Delta
  e^{j(l+0.5)\theta_n}: l=0,1,\ldots,2n\}
\end{equation}
where $N\ge2$ is a positive integer, $\Delta$ is a positive real
number and $\theta_n=2\pi/(2n+1)$. The set $A_N$ contains $N^2$ points
including the origin and $(2n+1)$ equally-spaced points on a circle of radius
$(n+0.5)\Delta$ for $n=1,2,\ldots,N-1$. The 9 points of $A_3$ are shown in
Fig. \ref{fig:A3} for illustration.

\begin{figure}[t!]
  \centering
  \begin{tikzpicture}[scale=1]
\def\PI{3.14}
    \fill[black!5] (0,0) circle (1);
    \node at (0,0) {X};
    \foreach \l/\c in {0/20,1/30,2/40} {
        \fill[black!\c] ({\l*120}:1) arc ({\l*120}:{(\l+1)*120}:1)
                 -- ++({(\l+1)*120}:1)  arc ({(\l+1)*120}:{\l*120}:2) -- ({\l*120}:1);
        \node at ({1.5*cos(2*\PI/3*(\l+0.5) r)},{1.5*sin(2*\PI/3*(\l+0.5) r)}) {X};
    }
    \draw[thin,dashed] (0,0) circle (1.5); 
    \foreach \l/\c in {0/45,1/35,2/25,3/15,4/5} {
        \fill[black!\c] ({\l*72}:2) arc ({\l*72}:{(\l+1)*72}:2)
                 -- ++({(\l+1)*72}:1)  arc ({(\l+1)*72}:{\l*72}:3) -- ({\l*72}:2);
        \node at ({2.5*cos(2*\PI/5*(\l+0.5)
          r)},{2.5*sin(2*\PI/5*(\l+0.5) r)}) {X};
    }
    \draw[thin,dashed] (0,0) circle (2.5); 
    \draw[<->] (-3.5,0) -- (3.5,0);
    \draw[<->] (0,-3.5) -- (0,3.5);
    \foreach \l/\s in {1/$\Delta$,2/$2\Delta$,3/$3\Delta$}
          \node[below] at (\l,0) {\s};
  \end{tikzpicture}
  \caption{The two-dimensional constellation $A_3$. X denotes a
    constellation point. For $x\in A_3$, $(X+U)|X=x$ is distributed
    uniformly in the shaded region around $x$.}
  \label{fig:A3}
\end{figure}
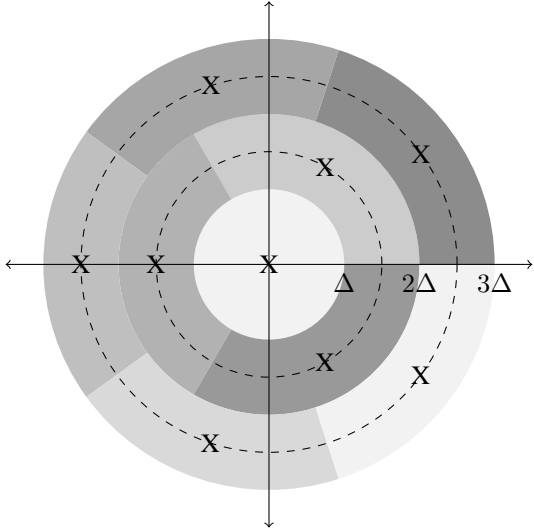

Define a random variable $U$ jointly distributed with $X\sim\text{Unif}(A_N)$ such that
$\tX=X+U$ is uniformly distributed in the circle of radius $N\Delta$
around the origin. For the constellation $A_3$, the distribution of
$\tX|X=x$ is illustrated through the shading around each point $x\in
A_3$ in Fig. \ref{fig:A3}. Specifically, for the constellation $A_N$, $U$ is defined
so that $(X+U)|X=0$ is uniform in the circle of radius $\Delta$ around
the origin, and $(X+U)|X=(n+0.5)\Delta e^{j(l+0.5)\theta_n}$ is
uniform in the region 
\begin{equation}
  \label{eq:53}
\{(r\cos\theta,r\sin\theta): n\le \frac{r}{\Delta}\le
n+1,l\le\frac{\theta}{\theta_n}\le l+1\}.  
\end{equation}
As before, the received value is
$Y=X+Z$. Since $\tX-X-Y$ forms a Markov chain, we have $I(X;Y)\ge I(\tX;Y)$ by
the data processing inequality. We further lower bound $I(\tX;Y)$ by using a strategy similar to the
real case in \cite{OzarowWyner90}. However, unlike the real case, $U$
and $X$ are correlated resulting in additional computations.

Since $I(\tX;Y)=h(\tX)-h(\tX|Y)$ and $h(\tX)=\log_2 \pi N^2\Delta^2$, we
lower bound $I(\tX;Y)$ by first upper bounding $h(\tX|Y)$ as follows:
\begin{align}
h(\tX|Y=y)&=-\int p(\tx|y)\log_2 p(\tx|y) d\tx\nonumber\\
&\le -\int p(\tx|y)\log_2 q_y(\tx) d\tx \label{eq:54}
\end{align}
where $q_y(\tx)$ is any valid density parametrized by $y$. Choosing
$q_y(\tx)=\frac{1}{2\pi s^2}e^{-||\tx-ky||^2/2s^2}$ (with parameters
$k$ and $s$ to be optimized later), we obtain
\begin{align}
\frac{h(\tX|Y=y)}{\log_2e}&\le \log 2\pi s^2+\frac{1}{2s^2}E[||\tX-ky||^2|Y=y]\nonumber\\
\frac{h(\tX|Y)}{\log_2e}&=\log 2\pi s^2+\frac{1}{2s^2}E[||\tX-kY||^2].  \label{eq:55}
\end{align}
Expanding using $\tX=X+U$ and $Y=X+Z$, and using the fact that $Z$ is
independent of $X$ and $U$, we have
\begin{equation}
  \label{eq:56}
  E[||\tX-kY||^2] = \frac{N^2\Delta^2}{2}-2(1+\rho_N)P_Nk+(P_N+2)k^2
\end{equation}
where $P_N\triangleq E[X^2]$, $\rho_N\triangleq\re(E[X^*U])/P_N$, and we have
used $E[||\tX||^2]=\frac{N^2\Delta^2}{2}$. To obtain the lowest upper
bound for $E[||\tX-kY||^2]$, we set
$k=k^*=\frac{(1+\rho_N)P_N}{2+P_N}$ and
$s^2=\frac{1}{2}E[||\tX-k^*Y||^2]$. Simplifying, we have
\begin{equation}
  \label{eq:57}
  h(\tX|Y)\le \log_2 \left(\pi e\left[\frac{N^2\Delta^2}{2}-\frac{P_N^2(1+\rho_N)^2}{P_N+2}\right]\right).
\end{equation}
We continue to bound lower $I(\tX;Y)$ by
\begin{align}
&\log_2\pi N^2\Delta^2-\log_2\left(\pi e\left[\frac{N^2\Delta^2}{2}-\frac{P_N^2(1+\rho_N)^2}{P_N+2}\right]\right)\nonumber\\
&=\log_2N^2-\log_2\frac{e}{2}-\log_2\left(N^2-\frac{P^2_N(1+\rho_N)^2}{(1+P_N/2)\Delta^2}\right).  \label{eq:58}
\end{align}
Defining $\tC=\log_2(1+P_N/2)$, and setting
$N^2=\alpha 2^{\tC}=\alpha(1+P_N/2)$ and simplifying, we have
\begin{align}
  I(\tX;Y)&\ge \tC-\log_2\frac{e}{2}-\log_2\left(\frac{N^2}{\alpha}-\frac{P^2_N(1+\rho_N)^2}{N^2\Delta^2}\right).
  \label{eq:59}  
\end{align}
In Appendix C, we show that
$P_N=\frac{\Delta^2N^2}{2}\left(1-O(1/N^2)\right)$,
$-0.66\le\rho_NN^2/(1-O(1/N^2))\le-0.64$ and provide details of the simplification needed to obtain the
following lower bound:
\begin{equation}
  \label{eq:60}
  I(X;Y)\ge \tC-0.45-\log_2\left(1+\frac{1.82}{\alpha}+O\left(\frac{1}{N^2}\right)\right).
\end{equation}
We see that the gap to the average-power constrained capacity with a
finite constellation can be made as small as 0.45 by choosing a large
enough $\alpha$ at high rates (large $N$). For moderate $N$, choosing
$\alpha=4$ results in a gap of less than 1 to capacity. Finally,
the rate in \eqref{eq:58} is achieved at a peak-power constraint of
$|X|\le (N-0.5)\Delta$ or an equivalent $\text{SNR}=(N-0.5)^2\Delta^2/2$. This lets one compare the analytical
lower bound against the other bounds shown in Fig. \ref{fig:2D}.

\section{$n$-Dimensional AWGN Channel}
\label{sec:nD}
Consider next the $n$-dimensinal AWGN channel
\begin{align}
   \ul{Y} = \ul{X} + \ul{Z}
\end{align}
where $\ul{Z}=[Z_1,Z_2,\ldots,Z_n]$ has independent Gaussian entries
with mean $0$ and variance $1$.
The channel density in Cartesian coordinates with $\ul{x}=[x_1,\ldots,x_n]$ and
$\ul{y}=[y_1,\ldots,y_n]$ is
\begin{align}
   p_{\ul{Y}|\ul{X}}(\ul{y}|\ul{x}) = \frac{1}{(2 \pi)^{n/2}} e^{-\| \ul{y} - \ul{x} \|^2/2}.
\end{align}

The $n$-dimensional spherical coordinate system has a radial coordinate $r$
and $n-1$ angular coordinates $\phi_i$, $i=1,2,\ldots,n-1$, where the domain
of $\phi_i$, $i=1,2,\ldots,n-2$, is $[0,\pi)$, and the domain of $\phi_{n-1}$ is
$[0,2\pi)$.\footnote{See http://en.wikipedia.org/wiki/N-sphere.}
For a point $\ul{x}$ with spherical coordinates $\ul{x}=[r_x,\phi_{x,1},\ldots,\phi_{x,n-1}]$ we 
can compute the Cartesian coordinates $\ul{x}=[x_1,\ldots,x_n]$ via
\begin{align*}
  x_1 & = r_x \cos(\phi_{x,1}) \\
  x_2 & = r_x \sin(\phi_{x,1}) \cos(\phi_{x,2}) \\
  x_3 & = r_x \sin(\phi_{x,1}) \sin(\phi_{x,2}) \cos(\phi_{x,3}) \\
  \vdots \\
  x_{n-1} & = r_x \sin(\phi_{x,1}) \ldots \sin(\phi_{x,n-2}) \cos(\phi_{x,n-1}) \\
  x_{n} & = r_x \sin(\phi_{x,1}) \ldots \sin(\phi_{x,n-2}) \sin(\phi_{x,n-1}). 
\end{align*}
In spherical coordinates the channel density has a complex form due to the
many sine and cosine terms. However, by symmetry we may restrict
attention to points $\ul{x}$ with $\phi_{x,i}=0$ for all $i$. For such $\ul{x}$, the
channel density in $n$-dimensional spherical coordinates is simply 
\begin{align}
   p_{\ul{Y}|\ul{X}}(\ul{y}|\ul{x}) = \frac{1}{(2 \pi)^{n/2}} e^{-\left( r_y^2 + r_x^2 - 2 r_x r_y \cos \phi_{y,1} \right)/2}.
\end{align}

Consider now the $n$-dimensional amplitude constraint $\|\ul{X}\| \le A$ where $A>0$.
We choose the test density
\begin{align} \label{eq:test-density-nD}
   q_{\ul{Y}}(\ul{y}) = \left\{
   \begin{array}{ll}
      \frac{\beta}{\text{Vol}(A)}, & r_y \le A \\
      \frac{1-\beta}{k_n(A) (2\pi)^{n/2}} e^{-(r_y-A)^2/2}, & r_y>A
   \end{array} \right.
\end{align}
where $\text{Vol}(r)$ is the volume of an $n$-dimensional ball with
radius $r$ and $k_n(A)$ is a constant that ensures that $q_{\ul{Y}}(\cdot)$
is a density. Again, the test density is uniform in the interval $r_y\le A$ and
is a ``split and scaled" Gaussian density for $r_y> A$.
We have\footnote{See http://en.wikipedia.org/wiki/Volume\_of\_an\_n-ball .}
\begin{align}
   \text{Vol}(r) = \frac{\pi^{n/2}}{\Gamma\left( n/2 + 1 \right)} \: r^n
\end{align}
where $\Gamma(\cdot)$ is Euler's gamma function. To compute $k_n(A)$,
observe that we require
\begin{align}
   \frac{1}{k_n(A) (2\pi)^{n/2}} \int_A^{\infty} \int_{0}^{\pi} \cdots \int_{0}^{\pi} \int_{0}^{2\pi} e^{-(r-A)^2/2} \: dV
   =1
\end{align}
where
\begin{align}
   dV = r^{n-1} dr \left[ \prod_{i=1}^{n-2} \sin^{n-1-i}(\phi_i) \: d\phi_i \right] d\phi_{n-1}
\end{align}
is the spherical volume element in $n$ dimensions. We thus have
\begin{align}
   & k_n(A) = \frac{2}{2^{n/2} \, \Gamma\left(\frac{n}{2}\right)} \int_A^{\infty} e^{-(r-A)^2/2} \: r^{n-1} dr \nonumber\\
   &=\frac{2^{\frac{2-n}{2}}}{\Gamma\left(\frac{n}{2}\right)} \int_A^{\infty} e^{\frac{-(r-A)^2}{2}}
       \left(\sum_{i=0}^{n-1}\binom{n-1}{i}A^{n-1-i} (r-A)^i\right)
       dr \nonumber\\
&=\frac{2^{\frac{2-n}{2}}}{\Gamma\left(\frac{n}{2}\right)}\left[\sum_{i=0}^{n-1}\binom{n-1}{i}A^{n-1-i}\int_0^{\infty}r^ie^{-r^2/2}dr\right].\label{eq:42}
\end{align}
Using the standard integral 
$$\int_0^{\infty}x^ne^{-ax^2}dx=\frac{\Gamma\left(\frac{n+1}{2}\right)}{2a^{\frac{n+1}{2}}}$$
the expression for $k_n(A)$ in \eqref{eq:42} simplifies to
\begin{align}
  \label{eq:43}
  k_n(A)=\sum_{i=0}^{n-1}\binom{n-1}{i}\frac{\Gamma\left(\frac{n-i}{2}\right)}{2^{i/2}\,\Gamma\left(\frac{n}{2}\right)}\,A^i.
\end{align}
For example, for $n=1$ we have $k_1=1$, and for $n=2$ we have
$k_2=1+\sqrt{\pi/2} \, A$.

\subsection{McKellips-type Bound}
 The divergence in \eqref{eq:upper-bound} can be written as
\begin{align} 
   & D\left( p_{\ul{Y}|\ul{X}}(\cdot | \ul{x}) \| q_{\ul{Y}}(\cdot) \right)
      = -\frac{n}{2}\log(2 \pi e) - \E{\log q_{\ul{Y}}(\ul{Y})} \nonumber \\
   & = \log \frac{\text{Vol(A)}}{(2\pi e)^{n/2} \beta} - \E{\log\left( q_{\ul{Y}}(\ul{Y}) \frac{\text{Vol}(A)}{\beta} \right)} .
    \label{eq:n-bounda}
\end{align}
The expectation in the above equation can be simplified and written as
\begin{align}
   & \frac{2}{2^{\frac{n-1}{2}}\Gamma\left(\frac{n-1}{2}\right)} \int_0^{\pi} \sin^{n-2}(\phi_1) \: e^{-(r_x^2/2) \sin^2 \phi_1} \nonumber \\
   & \quad \left\{ \int_A^{\infty} \psi\left( z-r_x  \cos \phi_{1} \right) \right. \nonumber \\
   & \quad \left. \left[ \log \frac{(1-\beta) \text{Vol}(A)}{(2\pi)^{n/2} \beta\, k_n(A) }
      - \frac{(r_y-A)^2}{2} \right] r_y^{n-1} \: dr_y \right\} d\phi_1. \label{eq:n-boundb}
\end{align}
For $n\ge2$, define the functions
\begin{align}
&\tI_n(x)=\frac{2}{2^{\frac{n-1}{2}}\Gamma\left(\frac{n-1}{2}\right)\sqrt{2\pi}}\int_0^{\pi}e^{x\cos\phi}\,(\sin\phi)^{n-2}\,d\phi \label{eq:26}\\
&Q_n(x,A)=\int_A^{\infty}e^{-(z^2+x^2)/2}\,\tI_n(zx)\,z^{n-1}dz \label{eq:24}\\
&g_n(x,A)=\int_A^{\infty}\frac{(z-A)^2}{2} e^{-(z^2+x^2)/2}\,\tI_n(zx)\,z^{n-1}dz  \label{eq:25}\\
&\tg_n(x,A)=\int_A^{\infty}\left(\frac{n}{2}-\frac{(z-A)^2}{2}\right) e^{-(z^2+x^2)/2}\,\tI_n(zx)\,z^{n-1}dz.\label{eq:46} 
\end{align}
In terms of the above functions, we can write
\begin{align}
 & D\left( p_{\ul{Y}|\ul{X}}(\cdot | \ul{x}) \| q_{\ul{Y}}(\cdot)
   \right) = \log \frac{\text{Vol(A)}}{(2\pi e)^{n/2} \beta}\nonumber\\
&+\left(\log\frac{(2\pi e)^{n/2}\beta\,k_n(A)}{(1-\beta)\tvol(A)}\right)Q_n(x,A)-\tg_n(x,A). \label{eq:45}    
\end{align}
In Appendix A, we show that $\tg_n(x,A)$ is positive. We
make the second $\log$ term in \eqref{eq:45} zero by choosing
\begin{align}
   \beta = \frac{\text{Vol}(A)}{\text{Vol}(A) + (2 \pi e)^{n/2}
  k_n(A)} 
\label{eq:51}
\end{align}
and we obtain the bound
$$C \le \log \left( k_n(A) + \frac{\text{Vol}(A)}{(2\pi e)^{n/2}} \right).$$
We combine this result with the capacity under the
(weaker) average power constraint $\E{\|X\|^2} \le A^2$. Observe
that the complex noise has power $n$ so the corresponding SNR
is $P=A^2/n$. We thus have the bound
\begin{align} \label{eq:C-bound-n-simple}
   C \le \min \left\{ \log \left( k_n(\sqrt{nP}) + \frac{\text{Vol}\left(\sqrt{nP}\right)}{(2\pi e)^{n/2}}  \right), \frac{n}{2} \log\left( 1+P \right) \right\}
\end{align}
where we measure the rate per $n$-dimensional symbol.

\subsection{Refined Bound}
The refinement is similar to the 2-dimensional case. We rewrite $D\left(
  p_{\ul{Y}|\ul{X}}(\cdot | \ul{x}) \| q_{\ul{Y}}(\cdot) \right)$ in
\eqref{eq:45} as follows:
\begin{align}
&  D\left( p_{\ul{Y}|\ul{X}}(\cdot | \ul{x}) \| q_{\ul{Y}}(\cdot)
  \right) = \log \frac{\text{Vol}(A)}{(2\pi e)^{n/2} \beta}\nonumber\\
& +\left(\log\frac{(2\pi)^{n/2}k_n(A)\beta}{(1-\beta)\tvol(A)}\right)Q_n(x,A)+g_n(x,A)  \label{eq:23}
\end{align}
where we have used the relationship
$\tg_n(x,A)=\frac{n}{2}Q_n(x,A)-g_n(x,A)$. As shown in Appendix B, the functions $Q_n(x,A)$ and $g_n(x,A)$ are
both increasing in $x$. Hence, for a positive
$\log\frac{(2\pi)^{n/2}k_n(A)\beta}{(1-\beta)\tvol(A)}$, the RHS of
\eqref{eq:23} is maximized at $x=A$, and we obtain the bound
\begin{align}
&  D\left( p_{\ul{Y}|\ul{X}}(\cdot | \ul{x}) \| q_{\ul{Y}}(\cdot)
  \right) \le \log \frac{\text{Vol(A)}}{(2\pi e)^{n/2} \beta}\nonumber\\
&
  +\left(\log\frac{(2\pi)^{n/2}\beta k_n(A)}{(1-\beta)\tvol(A)}\right)Q_n(A,A)+g_n(A,A) \label{eq:36}  
\end{align}
provided that
\begin{equation}
\beta\ge\frac{\tvol(A)}{(2\pi)^{n/2}k_n(A)+\tvol(A)}.\label{eq:37}
\end{equation}
Rewriting \eqref{eq:36} as
\begin{align}
&  D\left( p_{\ul{Y}|\ul{X}}(\cdot | \ul{x}) \| q_{\ul{Y}}(\cdot)
  \right) \le Q_n(A,A)\log \frac{k_n(A)}{e^{n/2}}\nonumber\\
&+(1-Q_n(A,A)) \log\frac {\tvol(A)}{(2\pi
  e)^{n/2}}+g_n(A,A)\nonumber\\
&  -(1-Q_n(A,A))\log\beta-Q_n(A,A)\log(1-\beta) \label{eq:40}
\end{align}
we see that $\beta=1-Q_n(A,A)$ minimizes the RHS of \eqref{eq:40}.
However, from \eqref{eq:37} this choice of $\beta$ is valid
only if
\begin{align}
  1-Q_n(A,A)\ge\frac{\tvol(A)}{(2\pi)^{n/2}k_n(A)+\tvol(A)} \label{eq:38}
\end{align}
which requires $A<A^*_n$, where $A^*_n$ is the smallest
positive value that results in equality in \eqref{eq:38}. The value
$A^*_n$ can be determined numerically. Choosing
$P=A^2/n$ and $P^*_n=(A^*_n)^2/n$, we obtain the bound
\begin{align}
C\le
  &(1-\beta_n(P))\log\left(k_n(\sqrt{nP})\right)+\beta_n(P)\log\frac{\text{Vol}\left(\sqrt{nP}\right)}{(2\pi
    e)^{n/2}}\nonumber\\
&+H_e(\beta_n(P))-\tg(\sqrt{nP},\sqrt{nP}),\quad\quad P\le P^*_n \label{eq:39}
\end{align}
where $\beta_n(P)=1-Q_n(\sqrt{nP},\sqrt{nP})$.

\subsection{Volume-based Lower Bound}
To obtain a lower bound, we use the volume-based method introduced in
\cite{JA2014}. The capacity of the $n$-dimensional AWGN channel
$\ul{Y}=\ul{X}+\ul{Z}$ with the peak constraint $|\ul{X}|\le A$ is
\begin{equation}
  \label{eq:10}
  C=\lim_{m\to\infty}\frac{1}{m}\sup_{p(\ul{x}^m)\in \cP^m_n}I(\ul{X}^m;\ul{Y}^m)
\end{equation}
where $\cP^m_n$ is the set of all valid distributions satisfying the
$n$-dimensional peak-power constraint. Now consider
\begin{align}
\nonumber  \frac{1}{m}I(\ul{X}^m;\ul{Y}^m)&=\frac{1}{m}(h(\ul{Y}^m)-h(\ul{Z}^m))\\
&=\frac{1}{m}h(\ul{Y}^m)-\frac{n}{2}\log(2\pi e).\label{eq:4}
\end{align}
Using the entropy-power inequality, we have
\begin{align}
\sup_{p(\ul{x}^m)\in \cP^m_n} e^{\frac{2}{mn}h(\ul{Y}^m)}&\ge
\sup_{p(\ul{x}^m)\in\cP^m_n}e^{\frac{2}{mn}h(\ul{X}^m)}+e^{\frac{2}{mn}h(\ul{Z}^m)}\nonumber\\
&\ge e^{\frac{2}{mn}\log(\tvol(A))^m}+e^{\log(2\pi e)}\nonumber\\
&\ge (\tvol_n(A))^{2/n}+2\pi e.
\end{align}
Therefore, we have
\begin{align}
\sup_{p(\ul{x}^m)\in \cP^m_n}\frac{1}{m}h(\ul{Y}^m)\ge\frac{n}{2}\log((\tvol(A))^{2/n}+2\pi e).
\label{eq:77star}
\end{align}
Using \eqref{eq:77star} in \eqref{eq:4}, we have
\begin{align}
  C&\ge \frac{n}{2}\log\left(1+\frac{\tvol(A)^{2/n}}{2\pi e}\right)\nonumber\\
&=\log\left(O(A^{n-2})+\frac{\tvol(A)}{(2\pi e)^{n/2}}\right).  \label{eq:11}
\end{align}

Comparing \eqref{eq:11} with the McKellips-type upper bound in
\eqref{eq:C-bound-n-simple}, we see that two bounds meet as SNR tends
to infinity, and they differ by $O(A^{n-1})$ inside the
logarithm. Thus, for moderate SNR the McKellips-type bound could be
improved. Comparing with the refined bound is not straight-forward,
and a numerical comparison for $n=2$ is shown in Fig. \ref{fig:2D}. 
The case $n=4$ is interesting because coherent optical
communication with two polarizations results in a 4-dimensional signal space. The bounds for
$n=4$ are plotted in Fig. \ref{fig:4d}.
As expected, the volume-based lower bound meets the
McKellips-type bound at high SNR. The analytical refined bound
is valid for SNR less than $P^*_4\approx7.92$ dB. Numerical
evaluation of $\min_{\beta}\max_x D\left( p_{\ul{Y}|\ul{X}}(\cdot | \ul{x}) \| q_{\ul{Y}}(\cdot)
  \right)$ is seen to be close to the lower bound for moderate
  SNR. For lower SNR, the volume-based lower bound is not expected
  to be tight. 

\begin{figure}[t!]
  \centering
%
%
\begin{tikzpicture}

\begin{axis}[%
width=3in,
height=2.26829786480649in,
scale only axis,
xmin=-5,
xmax=30,
xlabel={SNR (dB)},
xmajorgrids,
ymin=0,
ymax=20,
ylabel={Capacity bounds},
ymajorgrids,
legend style={at={(0.03,0.97)},anchor=north west,draw=black,fill=white,legend cell align=left,font=\small}
]
\addplot [color=blue,solid,mark=o,mark size=2.5,mark options={solid},thick]
  table[row sep=crcr]{%
-5	2.5681055127746\\
-3	3.09012829424184\\
-1	3.682226855266\\
1	4.34217597970433\\
3	5.06563561053951\\
5	5.84694989941574\\
7	6.68000496840671\\
9	7.55893833655714\\
11	8.47860235830749\\
13	9.4347817421539\\
15	10.4242167150095\\
17	11.444494680545\\
19	12.4938647340464\\
21	13.5710172424894\\
23	14.674861699714\\
25	15.8043301757907\\
27	16.9582279592759\\
29	18.1351448962818\\
31	19.333430204287\\
33	20.5512223627165\\
35	21.7865172607328\\
};
\addlegendentry{McKellips};

\addplot [color=blue,dashed,thick]
  table[row sep=crcr]{%
-5	0.792818322326228\\
-3	1.1722078528907\\
-1	1.6868876504073\\
1	2.35127326938478\\
3	3.16536470982311\\
5	4.11474641721359\\
7	5.17562874712406\\
9	6.32160884782605\\
11	7.52878873408572\\
13	8.77811793472609\\
15	10.055615346701\\
17	11.3515598036098\\
19	12.6594249188838\\
21	13.9749269179118\\
23	15.2952943249808\\
25	16.6187504824256\\
27	17.9441630866419\\
29	19.2708131566433\\
31	20.5982452253145\\
33	21.926171184662\\
35	23.2544089604321\\
};
\addlegendentry{$\frac{n}{2}\log(1+P)$};

\addplot [color=blue,solid,mark=asterisk,mark size=2.5,mark options={solid},thick]
  table[row sep=crcr]{%
-5	0.883428522006152\\
-3	1.25193570445074\\
-1	1.74089777085116\\
1	2.36069003375364\\
3	3.10193719477868\\
5	3.94443017715178\\
7	4.86687240857007\\
};
\addlegendentry{Refined};

\addplot [color=green,solid,mark=square,mark size=2.5,mark options={solid},thick]
  table[row sep=crcr]{%
7	4.86687241120768\\
9	5.85071650763746\\
11	6.88098816474191\\
13	7.94608760791583\\
15	9.04891453254959\\
17	10.1917341683648\\
19	11.3683591679596\\
21	12.5732732044529\\
23	13.8016531434087\\
25	15.0493515310067\\
27	16.312851985811\\
29	17.5892084430832\\
31	18.8759770685141\\
33	20.1711472654891\\
35	21.4730760374037\\
};
\addlegendentry{Numerical}

\addplot [color=red,solid,very thick]
  table[row sep=crcr]{%
-5	0.43947253822714\\
-3	0.668559227454527\\
-1	0.998048263179633\\
1	1.4536077843757\\
3	2.05438648258939\\
5	2.80676111368566\\
7	3.70198053556319\\
9	4.71936438727277\\
11	5.8328334925088\\
13	7.01699767050501\\
15	8.25065231564489\\
17	9.51775454377826\\
19	10.8069259745938\\
21	12.1104288612115\\
23	13.4231422327072\\
25	14.7417354665218\\
27	16.0640662959193\\
29	17.3887665016762\\
31	18.7149661334127\\
33	20.0421136174639\\
35	21.3698598656049\\
};
\addlegendentry{Lower bound}

\end{axis}
\end{tikzpicture}%
  \caption{Capacity bounds for 4-dimensional AWGN channels. The rate
    units are bits per 4 dimensions.}
  \label{fig:4d}
\end{figure}
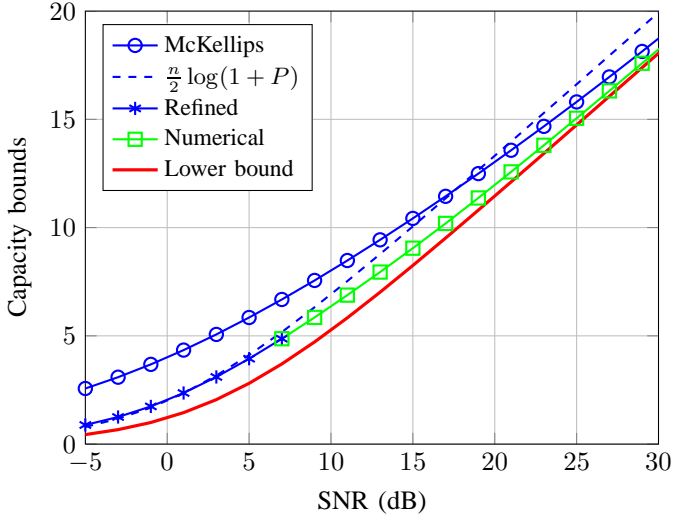

\subsection{Remarks}
Based on extensive numerical evaluations, we conjecture
that the expression for $D\left( p_{\ul{Y}|\ul{X}}(\cdot | \ul{x}) \| q_{\ul{Y}}(\cdot)
  \right)$ in \eqref{eq:23}, denoted as
\begin{align}
D_n(\beta,x)&\triangleq  \log \frac{\text{Vol}(A)}{(2\pi e)^{n/2} \beta}+g_n(x,A)\nonumber\\
& +\left(\log\frac{(2\pi)^{n/2}k_n(A)\beta}{(1-\beta)\tvol(A)}\right)Q_n(x,A)  \label{eq:41}
\end{align}
is maximized over $x\in[0,A]$ (for a fixed $\beta$) at the
endpoints $x=0$ or $x=A$. Under this conjecture, we have
\begin{equation}
  \label{eq:47}
  \min_{\beta}\max_xD_n(\beta,x)=\min_{\beta}\max(D_n(\beta,0),D_n(\beta,A)).
\end{equation}
Now, writing $D_n(\beta,x)$ as
\begin{align}
  D_n(\beta,x)&=-(1-Q_n(x,A))\log\beta-Q_n(x,A)\log(1-\beta) \nonumber \\
&+\text{ terms independent of }\beta  \label{eq:48}
\end{align}
we see that, for a fixed $x$ and $\beta\in[0,1]$, $D_n(\beta,x)$
achieves a minimum
at $\hat{\beta}_n(x)=1-Q_n(x,A)$, decreases with $\beta$ for
$\beta\in[0,\hat{\beta}_n(x)]$, and increases for $\beta\in[\hat{\beta}_n(x),1]$. Let
$\beta^*_n(A)$ be the value of $\beta$ for which $D_n(\beta,0)=D_n(\beta,A)$.

Using the above, the minmax in \eqref{eq:47} evaluates to
\begin{align}
  \min\{D_n(&\beta^*_n(A),A),\max(D_n(\hat{\beta}_n(0),0),D_n(\hat{\beta}_n(0),A)),\nonumber\\
&\max(D_n(\hat{\beta}_n(A),0),D_n(\hat{\beta}_n(A),A))\}. \label{eq:49}
\end{align}
To obtain an expression for $\beta^*_n(A)$, we
simplify $D_n(\beta,0)=D_n(\beta,A)$ resulting in
\begin{align}
  \label{eq:50}
\beta^*_n(A)=\frac{\tvol(A)}{\tvol(A)+(2\pi)^{n/2}e^{c_n(A)}k_n(A)}
\end{align}
where
$c_n(A)=\dfrac{g_n(A,A)-g_n(0,A)}{Q_n(0,A)-Q_n(A,A)}$. Interestingly,
as $A\to\infty$, $c_n(A)\to-1/2$, and we see that $\beta^*_n(A)$ tends
to the McKellips-type expression in \eqref{eq:51}.

Finally, we remark that the best upper bounds are as follows: 
\begin{enumerate}
\item at low SNR: the average-power capacity $\frac{n}{2}\log(1+P)$;
\item at moderate SNR: the refined upper bound which evaluates to
  $D_n(\hat{\beta}_n(A),A)$;
\item at high SNR: $D_n(\beta^*_n(A),A)$ which tends to the
  McKellips-type bound.
\end{enumerate}
The exact range of low, moderate, and high SNR depends on the
dimension $n$.

\section*{Acknowledgment}
G. Kramer and G. B\"ocherer were supported by the German Federal Ministry of Education and
Research in the framework of an Alexander von Humboldt Professorship.

\bibliographystyle{IEEEtran} 
\bibliography{muIT-Dec15,peakpower} 

\section*{Appendix A}
We show that the function $\tg_n(x,A)$ in \eqref{eq:46} is
positive. First, we rewrite $\tg_n(x,A)$ as
\begin{align}
&\tg_n(x,A)= \frac{1}{2^{\frac{n-1}{2}}\Gamma\left(\frac{n-1}{2}\right)} \int_0^{\pi} \sin^{n-2}(\phi) \: e^{-(x^2/2) \sin^2 \phi} \nonumber \\
   & \quad \left\{ \int_A^{\infty} \left[ n - (z-A)^2 \right] \psi\left( z-x  \cos \phi \right)  z^{n-1} \: dz \right\} d\phi. \label{eq:29}
\end{align}
Now it suffices to show that the integral
over $z$ above is positive. We set
$\tilde{z}=z-x\cos\phi$, $u=A-x\cos\phi$ and simplify the integral as
\begin{align}
   \int_u^{\infty} \left[ n - (\tilde{z}-u)^2 \right] (\tilde{z}-u+A)^{n-1} \psi\left(\tilde{z}\right) \: d\tilde{z} .
   \label{eq:n-boundc}
\end{align}
For example, for $n=2$ we recover \eqref{eq:D-bound2c}. Now
\eqref{eq:n-boundc} can be rewritten as
\begin{equation}
  \label{eq:3}
  \sum_{i=0}^{n-1}\binom{n-1}{i}A^{n-1-i}\int_u^{\infty}[n(\tilde{z}-u)^i-(\tilde{z}-u)^{i+2}]\psi(\tilde{z})d\tilde{z}.
\end{equation}
We claim that the integral inside the summation above is positive. In
fact, we prove the following stronger inequality for $u\ge 0$ and a non-negative integer $i$:
\begin{equation}
  \label{eq:1}
  \int_u^{\infty}[(i+1)(z-u)^i-(z-u)^{i+2}]\psi(z)dz\ge0.
\end{equation}
Since $n\ge i+1$ in \eqref{eq:3}, the inequality \eqref{eq:1} implies that the integral in
\eqref{eq:3} is positive.
A proof of \eqref{eq:1} is as follows:
\begin{align*}
&\int_u^{\infty}(z-u)^i\psi(z)dz=\int_u^{\infty}(z-u)^i\psi(z)(z-u)'dz\\
&\overset{(a)}{=}\left.(z-u)^{i+1}\psi(z)\right|_u^{\infty}-\int_u^{\infty}(z-u)[(z-u)^i(-z\psi(z))\\
&\phantom{=\int_u^{\infty}(z-u)^{i+1}(z-u+u)}+i(z-u)^{i-1}\psi(z)]dz\\
&=\int_u^{\infty}(z-u)^{i+2}\psi(z)dz+u\int_u^{\infty}(z-u)^{i+1}\psi(z)dz\\
&\phantom{=\int_u^{\infty}(z-u)^{i+1}(z-u+u)}-i\int_u^{\infty}(z-u)^i\psi(z)dz
\end{align*}
where $(a)$ uses integration by parts. The above simplifies to
\eqref{eq:1} because $u\int_u^{\infty}(z-u)^{i+1}\psi(z)dz\ge0$.


\section*{Appendix B}
We use the double factorial notation
\begin{align*}
  n!!=\begin{cases}
2 \cdot 4 \cdots n, &n:\text{even},\\
1 \cdot 3 \cdots n, &n:\text{odd}.
\end{cases}
\end{align*}
Using $e^{x\cos\phi}=\sum_{m=0}^{\infty}\frac{(x\cos\phi)^m}{m!}$ in
the integral of \eqref{eq:26}, we get the Taylor
series expansion for $\tI_n(x)$, $n\ge2$, as 
\begin{equation}
  \label{eq:27}
  \tI_n(x)=\frac{2^{(2-n)/2}}{\Gamma\left(\frac{n-1}{2}\right)\sqrt{\pi}}\sum_{m=0}^{\infty}a_{n-2,m}\frac{x^{m}}{m!}
\end{equation}
where $a_{n,m}=\int_{0}^{\pi}\sin^n\phi\cos^m\phi\,d\phi$. For
odd $m$, it is easy to see that $a_{n,m}=0$. For even $m$, 
we have
\begin{equation}
a_{n,m}=\begin{cases}
\dfrac{\pi\,(m-1)!!\,(n-1)!!}{(n+m)!!}, & n:\text{even},\\
\dfrac{2\,(m-1)!!\,(n-1)!!}{(n+m)!!}, & n:\text{odd}.
\end{cases}   \label{eq:28}
\end{equation}
Using the above $a_{m,n}$ in \eqref{eq:27} and simplifying, we have
\begin{align}
  \label{eq:31}
\tI_n(x)= d_n \sum_{k=0}^{\infty}\frac{(2k-1)!!}{(n+2k-2)!!}\frac{x^{2k}}{(2k)!}
\end{align}
with a suitably-defined $d_n$.

We show that the derivatives of the functions $Q_n(x,A)$ and
$g_n(x,A)$ in \eqref{eq:24} and \eqref{eq:25} with respect to $x$ are
non-negative. First, the derivative of $e^{-x^2/2}\tI_n(zx)$ with respect to
$x$ in Taylor series form is seen to be
\begin{align}
&  \frac{d}{dx}(e^{-x^2/2}\tI_n(zx)) \nonumber \\
&=d_ne^{-x^2/2}\sum_{k=1}^{\infty}\frac{(2k-3)!!}{(n+2k-2)!!}\frac{z^{2k}x^{2k-1}}{(2k-2)!}\nonumber\\
& \quad -d_ne^{-x^2/2}\sum_{k=0}^{\infty}\frac{(2k-1)!!}{(n+2k-2)!!}\frac{z^{2k}x^{2k+1}}{(2k)!}\nonumber\\
&=  d_ne^{-x^2/2}\sum_{k=0}^{\infty}\left(\frac{z^{2k+2}}{n+2k}-z^{2k}\right) \frac{(2k-1)!!}{(n+2k-2)!!}\frac{x^{2k+1}}{(2k)!}. \label{eq:30}
\end{align}
Using \eqref{eq:30} in the definition of $Q_n(x,A)$, we see that 
\begin{align}
&\frac{dQ_n(x,A)}{dx} \nonumber \\
&=\sum_{k=0}^{\infty}t_{n,k}(x)\int_A^{\infty}\left(\frac{z^{2k+2}}{n+2k}-z^{2k}\right)z^{n-1}e^{-z^2/2}dz \label{eq:34}
\end{align}
where $t_{n,k}(x)\triangleq
d_ne^{-x^2/2}\frac{(2k-1)!!}{(n+2k-2)!!}\frac{x^{2k+1}}{(2k)!}$. We
now show that the integral in \eqref{eq:30} is non-negative,
which implies that the derivative of $Q_n(x,A)$ with respect
to $x$ is non-negative. We have
\begin{align}
&\int_A^{\infty}(z^{2k})z^{n-1}e^{-z^2/2}dz=\int_A^{\infty}e^{-z^2/2}\left(\frac{z^{n+2k}}{n+2k}\right)'dz\nonumber\\
&=\frac{-A^{n+2k}e^{-A^2/2}}{n+2k}+\int_A^{\infty}\left(\frac{z^{2k+2}}{n+2k}\right)z^{n-1}e^{-z^2/2}dz \label{eq:35}
\end{align}
where we used integration by parts in the last step. Rearranging
the above equation, since $\frac{A^{n+2k}e^{-A^2/2}}{n+2k}>0$, we see that the integral in \eqref{eq:34} is
non-negative.

For $g_n(x,A)$, a similar approach is used. Using \eqref{eq:30} in the definition of $g_n(x,A)$, we see that 
\begin{align}
 &\frac{dg_n(x,A)}{dx}=\nonumber\\
&\sum_{k=0}^{\infty}t_{n,k}(x)\int_A^{\infty}\left(\frac{z^{2k+2}}{n+2k}-z^{2k}\right)\frac{(z-A)^2}{2}z^{n-1}e^{-z^2/2}dz.\label{eq:32}
\end{align}
A proof for the non-negativity of the integral in \eqref{eq:32} is
\begin{align}
\int_A^{\infty}(z^{2k})&\frac{(z-A)^2}{2}z^{n-1}e^{-z^2/2}dz\nonumber\\
&=\int_A^{\infty}\frac{(z-A)^2}{2}e^{-z^2/2}\left(\frac{z^{n+2k}}{n+2k}\right)'dz\nonumber\\
&=\int_A^{\infty}\left(\frac{z^{2k+2}}{n+2k}\right)\frac{(z-A)^2}{2}z^{n-1}e^{-z^2/2}dz\nonumber\\
&\quad -\int_A^{\infty}\left(\frac{z^{n+2k}}{n+2k}\right)(z-A)e^{-z^2/2}dz \label{eq:33}  
\end{align}
where we used integration by parts in the last step. Rearranging
the above equation, since $\left(\frac{z^{n+2k}}{n+2k}\right)(z-A)e^{-z^2/2}>0$
for $z>A$, we see that the integral in \eqref{eq:32} is
non-negative.
\section*{Appendix C}
Since $X$ is uniform in $A_N$ and $|A_N|=N^2$, we have
\begin{align}
P_N&=E[||X||^2]=\frac{1}{N^2}\sum_{n=1}^{N-1}\sum_{l=0}^{2n}(n+0.5)^2\Delta^2\nonumber\\
&=\frac{\Delta^2}{N^2}\sum_{n=1}^{N-1}(2n+1)\left(n^2+n+\frac{1}{4}\right)\label{eq:61}\\
&=\frac{\Delta^2}{2}\left(N^2-\frac{1}{2}\left(1+\frac{1}{N^2}\right)\right).\label{eq:62}
\end{align}
To evaluate $\rho_N=E[X^*U]/P_N$, we consider
$$E[X^*(X+U)]=E[E[X^*(X+U)|X]].$$ 
Now $X^*(X+U)|X=0$ is zero
with probability 1, and $(X+U)|X=(n+0.5)\Delta e^{j(l+0.5)\theta_n}$
is uniform in the region specified in \eqref{eq:53}. So
$X^*(X+U)|X=(n+0.5)\Delta e^{j(l+0.5)\theta_n}$ is uniform in the
region
\begin{align}
\{(r\cos\theta,r\sin\theta): n(n+0.5)\le&\frac{r}{\Delta^2}\le(n+1)(n+0.5),\nonumber\\ -0.5\le&\frac{\theta}{\theta_n}\le0.5\}.\label{eq:63}
\end{align}
A calculation shows that
\begin{align}
  E[X^*(X+U)|&X=(n+0.5)\Delta e^{j(l+0.5)\theta_n}]\nonumber\\
&=\Delta^2\left(n^2+n+\frac{1}{3}\right)\sinc\frac{\pi}{2n+1} \label{eq:64}
\end{align}
where $\sinc\ x = \frac{\sin x}{x}$. Therefore, we have
\begin{align}
E&[X^*(X+U)] \nonumber \\
&=\frac{1}{N^2}\sum_{n=1}^{N-1}\sum_{l=0}^{2n}\Delta^2\left(n^2+n+\frac{1}{3}\right)\sinc\frac{\pi}{2n+1}\nonumber\\
&=\frac{\Delta^2}{N^2}\sum_{n=1}^{N-1}(2n+1)\left(n^2+n+\frac{1}{3}\right)\sinc\frac{\pi}{2n+1}.  \label{eq:65}
\end{align}
Since $E[X^*(X+U)]=E[||X||^2]+E[X^*U]$, we have
\begin{align}
\rho_NP_N=\frac{\Delta^2}{N^2}\sum_{n=1}^{N-1}(2n+1)&\left[\left(n^2+n+\frac{1}{3}\right)\sinc\frac{\pi}{2n+1}\right.\nonumber\\
&\left.-\left(n^2+n+\frac{1}{4}\right)\right].\label{eq:66}
\end{align}
The sequence $a_n=(n^2+n+1/4)-(n^2+n+1/3)\sinc(\pi/(2n+1))$ is
increasing and converges to $\frac{\pi^2-2}{24}\le0.33$ with
$a_1\ge0.32$. We thus have
\begin{align}
  \label{eq:67}
  -0.33\left(1-\frac{1}{N^2}\right)\le\frac{\rho_NP_N}{\Delta^2}\le-0.32\left(1-\frac{1}{N^2}\right).
\end{align}
From \eqref{eq:62}, using $P_N=\frac{\Delta^2N^2}{2}(1-O(1/N^2))$, we have
\begin{align}
  \label{eq:68}
  -0.66\left(1-O\left(1/N^2\right)\right)\le\rho_NN^2\le-0.64\left(1-O\left(1/N^2\right)\right).  
\end{align}
The expression inside the second $\log$ term in \eqref{eq:59} is
$$\frac{N^2}{\alpha}-\frac{P^2_N(1+\rho_N)^2}{N^2\Delta^2}.$$
Since $P_N=\frac{\Delta^2N^2}{2}(1-0.5/N^2-O(1/N^4))$, we have
\begin{align*}
  \frac{P^2_N}{N^2\Delta^2}=P_N\frac{P_N}{\Delta^2N^2}=\left(\frac{N^2}{\alpha}-1\right)\left(1-\frac{0.5}{N^2}-O(1/N^4)\right)
\end{align*}
where we have also used $N^2=\alpha(1+P_N/2)$. We arrive at
\begin{align*}
  \frac{N^2}{\alpha}&-\frac{P^2_N(1+\rho_N)^2}{N^2\Delta^2}\\
&=\frac{N^2}{\alpha}-\left(\frac{N^2}{\alpha}-1\right)(1-0.5/N^2-O(1/N^4))(1+\rho_N)^2\\
&\ge 1+\frac{1.82}{\alpha}+O(1/N^2)
\end{align*}
using $\rho_NN^2\ge -0.66(1-O(1/N^2))$. This explains the final
analytical lower bound given in \eqref{eq:60}.
\end{document}